%% file: Main_File.tex
\documentclass[conference,letterpaper]{IEEEtran}

\input{Preamble}

\title{Rate-Distortion-Perception Function of Bernoulli Vector Sources}

 \author{\IEEEauthorblockN{Praneeth Kumar Vippathalla, Mihai-Alin Badiu and Justin P. Coon}
 \IEEEauthorblockA{Department of Engineering Science\\ University of Oxford\\ OX1 3PJ Oxford, U.K. \\Email: \{praneeth.vippathalla,  mihai.badiu, justin.coon\}@eng.ox.ac.uk}
\thanks{This research was funded in whole or in part by the U. S. Army Research Laboratory and the U. S. Army Research Office (W911NF-22-1-0070 and W911NF-24-2-0102). For the purpose of Open Access, the authors have applied a CC BY public copyright licence to any Author Accepted Manuscript (AAM) version arising from this submission.}}

\begin{document}
\IEEEoverridecommandlockouts
\maketitle

\begin{abstract}
  \input{Abstract}
\end{abstract} 

\input{Introduction}

\input{Problem_formulation}
\input{Results}

\input{Proof}

\input{RDP_graphs}

\input{Conclusion}

\bibliographystyle{IEEEtran}
\bibliography{IEEEabrv,References}


\end{document}

%% file: Preamble.tex
\usepackage{tikz}
\usetikzlibrary{shapes, arrows, positioning, calc}
\usepackage{graphicx,xcolor}

\usepackage{amsmath,amssymb,amsthm,amsfonts}
\usepackage{bbm}
\usepackage{bm}
\usepackage{dsfont}

\newtheorem{theorem}{Theorem}

\newtheorem{lemma}{Lemma}
\newtheorem{corollary}{Corollary}
\usepackage{cleveref}

\usepackage[noadjust]{cite}
\allowdisplaybreaks

\DeclareMathOperator*{\suchthat}{s.t.}

%% file: Abstract.tex
In this paper, we consider the rate-distortion-perception (RDP) trade-off for the lossy compression of a Bernoulli vector source, which is a finite collection of independent binary random variables. The RDP function quantifies in a way the efficient compression of a source when we impose a distortion constraint that limits the dissimilarity between the source and the reconstruction and a perception constraint that restricts the distributional discrepancy of the source and the reconstruction. In this work, we obtain an exact characterization of the RDP function of a Bernoulli vector source with the Hamming distortion function and a single-letter perception function that measures the closeness of the distributions of the components of the source. The solution can be described by partitioning the set of distortion and perception levels $(D,P)$ into three regions, where in each region the optimal distortion and perception levels we allot to the components have a similar nature. Finally, we introduce the RDP function for graph sources and apply our result to the Erd\H{o}s-R\'enyi graph model.

%% file: Introduction.tex
\section{Introduction}\label{sec:introduction}

The rate-distortion-perception (RDP) function, introduced by Blau and Michaeli in \cite{blau19}, is used to study the lossy compression with an additional constraint on the perceptual quality of the reconstruction. In image compression, sometimes though a reconstructed image is well within a desired mean squared error limit from the original image, the visual quality of the reconstruction could be quite poor. In order to have a better reconstruction, a constraint on a measure of perception has been used to study the lossy compression, extending the classical rate-distortion framework. This measure is typically defined as a distance between probability distributions.

Let $\mathcal{X}$ be the (finite) alphabet of a source and its reconstructed version. Given a distortion function $\Delta:\mathcal{X} \times \mathcal{X} \to [0, \infty)$ and a perception function $d:\mathcal{P}({\mathcal{X}})\times \mathcal{P}({\mathcal{X}}) \to [0, \infty)$ with $\mathcal{P}({\mathcal{X}})$ being the set of all probability distributions defined on $\mathcal{X}$, the rate-distortion-perception function of a random source $X \sim P_{X} \in \mathcal{P}({\mathcal{X}}) $ at a distortion level $D$ and a perception level $P$  is defined as
\begin{align}\label{eq:rdp:def}
    R(D,P)= &\min_{P_{\hat{X}|X}} I(X;\hat{X}) \\
            &\suchthat \  \mathbb{E}\left[\Delta(X;\hat{X})\right]\leq D, \ d(P_X, P_{\hat{X}})\leq P.\nonumber
\end{align}
The perception function $d$ measures how far the distributions $P_X$ and $P_{\hat{X}}$ are in some sense. 

Coding arguments also exist for the RDP function \cite{RDP_function_chen, theis2021a}. In particular, Theis and Wagner \cite{theis2021a} gave a coding argument for a one-shot setting using a stochastic encoder and decoder with the variable-length coding. Let $U \in \mathbb{R}$ be the randomness (independent of $X$) shared by the encoder and decoder. The encoder outputs a finite length binary string $K \in \{0,1\}^*$ with the prefix property, using inputs $X$ and $U$. The decoder outputs $\hat{X} \in \hat{\mathcal{X}}$ using $K$ and $U$ so that $\hat{X}$ is a reconstructed version of $X$ satisfying the distortion and perception constraints $\mathbb{E}\left[\Delta(X;\hat{X})\right]\leq D$ and $\ d(P_X, P_{\hat{X}})\leq P$. The goal is to minimize the average length of $K$ over all encoder and decoder pairs, and the shared randomness $U$, i.e.,
$L\triangleq\min\ \mathbb{E}[l(K)],$
where the minimization is over all the encoder and decoder pairs and all the random variables $U$, with $l(\cdot)$ denoting the length function. The optimal average length and the RDP function are related as follows:
  \begin{align}\label{eq:conv:theis}
        &R(D,P)\leq L \leq R(D,P) + \log (R(D,P) + 1) + 5
    \end{align}
The converse readily follows from conditioning on a realization of the private randomness and appropriately bounding the entropies. For the achievability part, \cite{theis2021a} used the strong functional representation lemma of \cite{cheuk18} to show the existence of a shared random variable $U$, an encoder and a decoder such that the average length of the encoder's output $K$ is less than $R(D,P) + \log (R(D,P) + 1) + 5$.

Recently, there have been several efforts to characterize the RDP tradeoff for different sources and models \cite{RDP_vec_gaussian,computation_RDP_gaussian, dror_24, yassine_23}. For instance, \cite{RDP_vec_gaussian} studies the RDP tradeoff for a Gaussian vector source under different perception measures. In this work, we consider a Bernoulli vector source and exactly characterize its RDP function for the Hamming distortion function and a perception function that measures the closeness of the distributions of the components of the source. Finally, we introduce the RDP function for graph sources and apply the obtained characterization to the case of Erd\H{o}s-R\'enyi graphs.

%% file: Results.tex
\section{RDP Function of a Bernoulli Vector Source}

A Bernoulli vector source is a collection $\mathbf{X}=\{X_i:i \in [n]\}$ of independent binary random variables with respective success probabilities $\{q_i:i \in [n]\}$. We are interested in finding the RDP function \eqref{eq:rdp:def} of a Bernoulli vector source with the Hamming distortion measure $\Delta(\mathbf{X};\hat{\mathbf{X}})= \sum_{i}\mathds{1}(X_i \neq \hat{X}_i)$ and the perception measure $d(P_{\mathbf{X}},P_{\hat{\mathbf{X}}})= \sum_{i}|\mathbb{P}(X_{i}=1)-\mathbb{P}(\hat{X}_{i}=1)|$. For simplicity, we are considering a single-letter perception measure that quantifies how far the respective marginal distributions of the components of the source and the reconstructed version are. 

Without loss of generality, we can consider the case of $0\leq q_i \leq 1/2$ for all $i \in [n]$. If $q_j \geq 1/2$, we can work with $X_j \oplus 1$ instead of $X_j$ to reduce it to the case of $1-q_j$, which is less than $1/2$. We also assume without loss of generality that $q_1\geq q_2 \geq \ldots \geq q_n$. When $i=1$, the RDP function is just that of a Bernoulli source $X\sim \text{Bern}(q)$, $q\leq 1/2$, which was computed in \cite{blau19}. For $P \leq q$, it is given by
\begin{align}\label{rdpber1}
    &R(D,P,q)\nonumber\\&= \begin{cases}
        h_2(q) - h_2(D) \qquad \text{if} \ 0\leq D < \frac{P}{1-2(q-P)}, \\
        0 \mkern 134mu \text{if} \ D \geq 2q(1-q)-(1-2q)P, \\
        2h_2(q) + h_2(q-P) - h_3\left(\frac{D-P}{2}, q\right)  - h_3\left(\frac{D+P}{2}, 1-q\right) \\
        \mkern 140mu \text{otherwise},
    \end{cases}
\end{align}
and for $P \geq q$, the function is given by
\begin{align}\label{rdpber2}
    R(D,P,q)= \begin{cases}
        h_2(q) - h_2(D) & \text{if} \ 0\leq D < q, \\
        0 & \text{otherwise},
    \end{cases}
\end{align}
where $h_2(u):= -u \log u - (1-u) \log (1-u)$ is the binary entropy function of the probability vector $(u,1-u)$, and $h_3(u, v):= -u \log u - v \log v -(1-u-v) \log (1-u-v)$ is the ternary entropy function of the probability vector $(u,v, 1-u-v)$. The function\footnote{Note the usage of the parameter $q$ in the notation of the RDP function of a Bernoulli source to distinguish it from a vector source.} $R(D,P,q)$ is convex and non-increasing in $(D,P)$.                       

\subsection{Main Results}

The following theorem expresses the RDP function of a Bernoulli vector source in terms of its components.
\begin{theorem}\label{thm:ber:vec}
For a Bernoulli vector source with parameters $q_i \leq 1/2$, $i \in [n]$, the rate-distortion-perception function is given by 
    \begin{align}\label{eq:rdp_qi}
    &R(D,P)= \min\bigg\{\sum_{i}R(d_{i}, p_{i}, q_{i}): 0\leq d_i,\ 0 \leq  p_i, \nonumber\\
    &\mkern 200mu \sum_{i}d_{i} \leq D,\ 
      \sum_{i} p_{i}\leq P\bigg\}
    \end{align}
    where $R(d,p,q)$ is the RDP function of a Bernoulli source with probability $q$, which is as defined in \eqref{rdpber1} and \eqref{rdpber2}.
\end{theorem}
\begin{proof}
For a conditional distribution $P_{\hat{\mathbf{X}}|\mathbf{X}}$ satisfying the constraints $\sum_i \mathbb{P}(X_{i}\neq \hat{X}_{i}) \leq D$ and $\sum_{i}|\mathbb{P}(X_{i}=1)-\mathbb{P}(\hat{X}_{i}=1)|\leq P$, we have
    \begin{align}
     &I(\mathbf{X}; \hat{\mathbf{X}})\nonumber\\
     &= H(\mathbf{X}) -H(\mathbf{X}| \hat{\mathbf{X}})\nonumber\\
     & = \sum_{i} H\left(X_i\right) - \sum_{i} H(X_i\big| \{\hat{X}_j:j<i\}) \nonumber \\
     & \geq \sum_{i} H\left(X_i\right) - \sum_{i} H(X_i\big| \hat{X}_i) \nonumber \\
     & = \sum_{i} I(X_i;\hat{X}_i)\nonumber\\
     & \geq \sum_{i}R(d_{i}, p_{i}, q_{i}) \label{eq:di_pi_qi} \\
     & \geq \min\bigg\{\sum_{i}R(d_{i}, p_{i}, q_{i}): 0\leq d_i,\ 0 \leq  p_i,\ \sum_{i}d_{i} \leq D,\nonumber\\
    &\mkern 300mu \  \sum_{i} p_{i}\leq P\bigg\}\label{eq:mon2} 
    \end{align}
where \eqref{eq:di_pi_qi} uses the definition of the RDP function of a Bernoulli source with the notation that $d_{i}:=\mathbb{P}(X_i\neq \hat{X}_i)$ and $p_{i}:=|\mathbb{P}(X_{i}=1)-\mathbb{P}(\hat{X}_{i}=1)|$, which satisfy the distortion and perception constraints yielding \eqref{eq:mon2}.

With $d^*_{i}$'s and $p^*_{i}$'s being the optimizers in \eqref{eq:mon2}, the lower bound can be achieved by choosing a distribution of the form
\begin{align}  \label{eq:ind_form} P_{\hat{\mathbf{X}}|\mathbf{X}}&= \prod_{i}P_{\hat{X}_i|X_i},
 \end{align}
 where the conditional distributions $P_{\hat{X}_i|X_i}$ achieve the rates $R(d^*_{i}, p^*_{i}, q_{i})$ of a Bernoulli source. 
\end{proof}
The arguments used in the above theorem work in general for single-letter distortion and perception measures.

\begin{corollary}
If $q_i=q$, then
    $$R(D,P)=n R\left(D/n, P/n, q\right),$$
    where $R(D,P,q)$ is the RDP function of a Bernoulli random var with probability $q$, which is as defined in \eqref{rdpber1} and \eqref{rdpber2}.
\end{corollary}
\begin{proof}
The result immediately follows from Theorem~\ref{thm:ber:vec} by using the convexity and monotonicity of the function $R(d,p,q)$ in arguments $d$ and $p$:
    \begin{align*}
    \sum_{i}R(d_{i}, p_{i}, q)& \geq n R\bigg({\sum_{i}d_{i}}/{n}, {\sum_{i}p_{i}}/{n}, q\bigg),\\
     & \geq n R\left({D}/{n}, {P}/{n}, q\right).
    \end{align*}
\end{proof}
We will now solve the optimization of Theorem~\ref{thm:ber:vec} by considering a partition of $\mathbb{R}^2_{+}:=\{(D,P): D\geq 0, P\geq 0\}$. For this purpose, let us define two functions in $D$:
\begin{align}
    T(D) \triangleq\sum_{i} \frac{2d_{i}(1-q_{i})}{1-2d_{i}},
\end{align}
for $ 0\leq D < \sum_i q_{i}$,
where $d_{i}=\min\left\{\frac{1}{e^{\beta}+1}, q_{i}\right\}$ with $\beta\geq 0$ chosen such that $\sum_{i}d_{i}=D$, and
\begin{align}
    S(D) \triangleq &\min \sum_i p_{i}\\ &\suchthat \ q_{i}\leq d_{i},\  0\leq p_{i},\ \sum_{i}d_{i}= D,\nonumber\\
    & \mkern 40mu 2q_{i}(1-q_{i})-(1-2q_{i})p_{i}\leq d_{i}. \nonumber
\end{align}
for $ \sum_i q_{i}\leq D$. The next lemma, which is stated without proof, characterizes $S(D)$ and can readily be proved by solving the corresponding KKT conditions.
\begin{lemma}\label{lem:partition_D_P}
For $ \sum_iq_{i}\leq D \leq \sum_i2q_{i}(1-q_{i})$, we have
\begin{align}
    S(D)= \frac{2q_k(1-q_k)-d_k}{1-2q_k}+\sum_{i=k+1}^{n}q_i,
\end{align}
 where $k$ and $d_k$ are chosen such that $\sum_{i=1}^{k-1}2q_i(1-q_i) +d_k +\sum_{i=k+1}^{n}q_i=D$. Here the optimizers are
 \begin{align}
     (d_i^*,p_i^*)=\begin{cases}
         \left(2q_i(1-q_i), 0\right) & \text{if } i \leq k-1\\
         \left(d_k, \frac{2q_k(1-q_k)-d_k}{1-2q_k}\right) & \text{if } i = k\\
         \left(q_i,q_i\right)& \text{if } k+1 \leq i \leq n.\\
     \end{cases}
 \end{align}
For $\sum_i2q_{i}(1-q_{i})\leq D,$
\begin{align}
    S(D)=0  
\end{align}
where the optimal value is achieved by pairs that satisfy $p_i^*=0$, $d_i^*\geq 2q_{i}(1-q_{i})$ and $\sum_{i}d_i^*=D$.
\end{lemma}
Define a partition of $\mathbb{R}^2_{+}$ as follows:
\begin{align}
    \mathcal{A} &\triangleq \left\{(D,P) : 0\leq D < \sum \limits_{i}q_{i},\  T(D) \leq P \right\},\\
    \mathcal{B} &\triangleq \left\{(D,P) :  \sum \limits_{i}q_{i}\leq D,\  S(D) \leq P \right\},\\
    \mathcal{C} &\triangleq \mathbb{R}^2_{+} \setminus \left(\mathcal{A} \cup \mathcal{C} \right),
\end{align}
which is illustrated in Fig.~\ref{fig:D_P}. 

 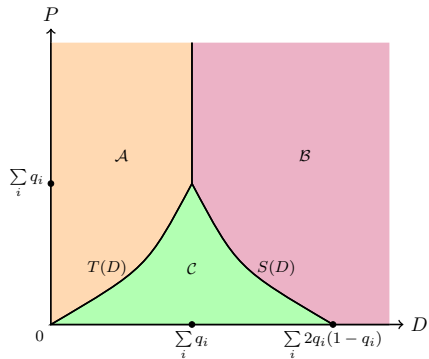
\begin{figure}[h]
\centering
\resizebox{0.75\width}{!}{\input{Figures/D_P_regions}}
\caption{A partition of $\mathbb{R}^2_{+}:=\{(D,P): D\geq 0, P\geq 0\}$.}
\label{fig:D_P}
 \end{figure}

The following theorem characterizes the RDP function of a Bernoulli vector source.
 \begin{theorem}\label{thm:closed_form}
For the RDP function of a Bernoulli vector source with parameters $q_i \leq 1/2$, $i \in [n]$, there exists optimizers $(d_i^*,p_i^*)$ satisfying the constraints with equality, i.e., $\sum_{i} d^*_{i}=D$ and $\sum_{i} p^*_{i}=P$. Moreover, the RDP function can be expressed as below. 
     \begin{enumerate}
         \item If $(D,P) \in \mathcal{A}$, then 
         \begin{align}
             R(D,P)= \sum_{i}\left[h_2(q_{i}) - h_2(d^*_{i})\right].
         \end{align}
         where $d^*_{i}=\min\left\{\frac{1}{e^{\beta}+1}, q_{i}\right\}$ with $\beta$ chosen such that $\sum_{i} d^*_{i}=D$ if $D>0$, and $d^*_i=0$ if $D=0$. In addition, 
         $p^*_{i}$'s can be chosen in any way that satisfy the constraints $p^*_{i}\geq \frac{2d^*_{i}(1-q_{i})}{1-2d^*_{i}}$ and $\sum_{i} p^*_{i}=P$.
         \item If $(D,P) \in \mathcal{B}$, then 
         $$R(D,P)=0,$$
         where the optimizers  $(d_i^*,p_i^*)$'s can be chosen in any way that satisfy the constraints  $q_{i}\leq d^*_{i},\  0\leq p^*_{i},\ 2q_{i}(1-q_{i})-(1-2q_{i})p^*_{i}\leq d^*_{i},$ $\sum_{i} d^*_{i}=D$ and 
         $\sum_{i} p^*_{i}=P$.
         \item If $(D,P) \in \mathcal{C}$, then 
         \begin{align}
             &R(D,P)\nonumber\\
             &\mkern 15mu= \sum_{i}\left[2h_2(q_{i}) + h_2(q_{i}-p^*_{i}) - h_3\left(\frac{d^*_{i}-p^*_{i}}{2}, q_{i}\right)  \right. \nonumber\\
             &\left. \mkern 162mu- h_3\left(\frac{d^*_{i}+p^*_{i}}{2}, 1-q_{i}\right)\right],
         \end{align}
         where   $(d^*_{i}, p^*_{i})$'s are given as follows. Set $d^*_{i}= d'_{i} \mathds{1}\{p'_{i}> 0\} + d''_i \mathds{1}\{p'_{i}\leq 0\}$ and $p^*_{i}= p'_{i} \mathds{1}\{p'_{i}> 0\}$, where $(d'_{i}, p'_{i})$ solves
         \begin{align}
        \alpha&=-\frac{1}{2}\log \frac{(d'_{i}-p'_{i})(d'_{i}+p'_{i})}{\left(2(1-q_{i})-(d'_{i}-p'_{i})\right)\left(2q_{i}-(d'_{i}+p'_{i})\right)}\\
        \beta&=-\frac{1}{2}\log \frac{(q_{i}-p'_{i})^2(d'_{i}+p'_{i})\left(2(1-q_{i})-(d'_{i}-p'_{i})\right)}{(1-q_{i}+p'_{i})^2(d'_{i}-p'_{i})\left(2q_{i}-(d'_{i}+p'_{i})\right)}
    \end{align} and 
     $d''_i$ solves 
      \begin{align}\label{eq:p_zero}
        d''_{i} =             \frac{\sqrt{1+4q_{i}(1-q_{i})\left(e^{2\alpha}-1\right)}-1}{e^{2\alpha}-1} 
    \end{align}
    with $\alpha>0$ and $\beta>0$ chosen such that $\sum_{i} d^*_{i}=D$ and $\sum_{i} p^*_{i}=P$.
     \end{enumerate}
 \end{theorem}

 \begin{proof}
     See Section~\ref{sec:proof}.
 \end{proof}

 The above theorem says that if $(D,P) \in \mathcal{A}$, then the RDP function of a Bernoulli vector source is nothing but its rate-distortion function with the perception constraint not affecting the optimal value. The optimal solution is same as that of the rate-distortion function. On the other extreme, when $(D,P) \in \mathcal{B}$, we can have a reconstructed version that is independent of the source. The perception constraint plays a role when  $(D,P) \in \mathcal{C}$.

\begin{corollary}
    If the perception index is zero, i.e., $P=0$, then we have
    \begin{align*}
        R(D,0)= \sum_{i} \left[2h_2(q_{i}) - h_3 \left(\frac{d^*_{i}}{2}, q_{i}\right) - h_3 \left(\frac{d^*_{i}}{2}, 1-q_{i}\right)\right]^{+}
    \end{align*}
where $d_i^*$ solves \eqref{eq:p_zero}
 for $D <\sum_{i} 2q_{i}(1-q_{i})$ with $\alpha>0$ chosen such that $\sum_{i}  d^{*}_{i} = D$, and for $ \sum_{i} 2q_{i}(1-q_{i}) \leq D$, $d_i^*$'s are chosen such that 
$ 2q_{i}(1-q_{i})\leq D$ and $\sum_{i}  d^{*}_{i} = D$.
\end{corollary}
\begin{proof}
When $P=0$, $(D,P)$ can be either in $\mathcal{B}$ or $\mathcal{C}$. In both of these cases, $p_i^*=0$. Therefore, by applying Theorem~\ref{thm:closed_form}, the corollary follows.
\end{proof}

%% file: Figures/D_P_regions.tex
\begin{tikzpicture}[scale=10]

    \draw[thick, ->] (0,0) -- (0.625,0) node[right] {$D$};
    \draw[thick, ->] (0,0) -- (0,0.525) node[above] {$P$};
    
    \fill[purple!30] (0.25,0.5) -- (0.25,0.25) .. controls (0.33,0.1) .. (0.5,0) -- (0.6,0) -- (0.6,0.5)-- cycle;
    \fill[green!30] (0,0) .. controls (0.17,0.1) .. (0.25,0.25).. controls (0.33,0.1) .. (0.5,0)--cycle;
    
    \fill[orange!30] (0,0.5)--(0,0) .. controls (0.17,0.1) .. (0.25,0.25)--(0.25,0.5)--cycle;

   \node at (0.45,0.3) {\footnotesize$\mathcal{B}$};
   \node at (0.125,0.3) {\footnotesize$\mathcal{A}$};
\node at (0.25,0.1) {\footnotesize$\mathcal{C}$};
\node at (0.1,0.1) {\footnotesize$T(D)$};
\node at (0.4,0.1) {\footnotesize$S(D)$};

\draw [fill] (0.25,0) circle [radius=.005] node [below] {\footnotesize$\sum \limits_{i}q_{i}$};
\draw [fill] (0,0.25) circle [radius=.005] node [left] {\footnotesize$\sum \limits_{i}q_{i}$};
\draw [fill] (0.5,0) circle [radius=.005] node [below] {\footnotesize$\sum \limits_{i}2q_{i}(1-q_{i})$};
\node at (-0.02,-0.02){\footnotesize$0$};

    \draw[thick] (0.25,0.5) -- (0.25,0.25) .. controls (0.33,0.1) .. (0.5,0) -- (0.6,0);
    \draw[thick] (0,0) .. controls (0.17,0.1) .. (0.25,0.25).. controls (0.33,0.1) .. (0.5,0)--cycle;    
    \draw[thick] (0,0.5)--(0,0) .. controls (0.17,0.1) .. (0.25,0.25)--(0.25,0.5);
   
\end{tikzpicture}

%% file: Proof.tex
\section{Proof of Theorem~\ref{thm:closed_form}}\label{sec:proof}
Suppose $D=0$. Then, we necessarily have to set $d_i^*=0$ because the constraints in the optimization problem are $\sum_id_i \leq D$ and $d_i\geq 0, i \in [n]$.  This yields $R(0,P)= \sum_{i}h_2(q_i)$. Here $p_i^*$'s can be chosen in any way such that $\sum_{i}p^*_i=P$ and $p^*_i\geq 0$ for $i \in [n]$. For the rest of the proof, we assume that $D>0$.

Let us first argue that there exist optimizers $(d^*_{i},p^*_{i})$ satisfying $\sum_{i}d^*_{i}= D$ and $\sum_{i}p^*_{i}= P$. Suppose that one of the constraints is strict, i.e., $\sum_{i}d^*_{i}< D$ or $\sum_{i}p^*_{i}< P$. If $\sum_{i}d^*_{i}< D$ then we can increase the value of each $d^*_{i}$ to get $\tilde{d}_i\triangleq d^*_{i}+\epsilon$, where $\epsilon$ is chosen small enough that the resulting numbers satisfy the  constraint $\sum_{i}\tilde{d}_i= D$. As $\sum_{i}R(d_{i}, p_{i}, q_{i})$ is a non-increasing function in the arguments $d_{i}$'s and $p_{i}$'s, we have that 
    $R(D,P)=\sum_{i}R(d^*_{i}, p^*_{i}, q_{i})\geq \sum_{i}R(\tilde{d}_i, p^*_{i}, q_{i})\geq R(D,P)$, implying that there exist optimizers satisfying the equality $\sum_{i}d^*_{i}= D$. A similar argument can be used for $\sum_{i}p^*_{i}= P$.

With these simplifications, we will solve the optimization problem.  The corresponding Lagrangian function is 
\begin{align}\label{eq:subdiff}
    L= \sum_{i}R(d_{i},p_{i},q_{i}) +  \nu&\left(\sum_{i}d_{i} - D\right)+\mu\left(\sum_{i}p_{i} - P\right)\nonumber\\
    &-\sum_{i}\gamma_{i}d_i-\sum_{i}\lambda_ip_{i},
\end{align}
where $d_{i}\geq 0, p_{i}\geq 0,\gamma_i\geq 0,\lambda_i\geq 0$ for all $i \in [n]$ and $\nu, \mu \in \mathbb{R}$. As the optimization problem is convex and satisfies Slater's condition \cite{Boyd_Vandenberghe_2004}, strong duality holds, i.e., the optimal values of the primal and the dual problems are the same. It also means that the optimizers of the primal problem $d^*_{i}, p^*_{i}$ minimizes the Lagrangian function \eqref{eq:subdiff} evaluated at the optimizers of the dual problem $\gamma^*_i,\lambda^*_i, \nu^*, \mu^*$ for $i \in [n]$. Moreover, the optimizers satisfy the complementary slackness condition: for $i \in [n]$, $\gamma^*_i d^{*}_i= 0$ and $\lambda^*_i p^{*}_i= 0$.

Now we will argue that $d_i^*>0$ for every $i$. Suppose that there is an index $u$ for which $d_u^*=0$. As $D>0$, there must be another index $v \neq u$ such that $d_v^*>0$.  Because of the complementary slackness condition we have $\gamma_v^*=0$. For a small enough $\epsilon>0$, consider different distortion values $\tilde{d}_u=\epsilon$, $\tilde{d}_v=d_v^*-\epsilon$, and $\tilde{d}_j=d_j^*$ for $j \notin \{u,v\}$. These distortion values are still feasible as $\sum_i \tilde{d}_i=D$. Since $d^*$ minimizes the Lagrangian function, we have $L(d^*)\leq L(\tilde{d})$, where the optimizers $p_i^*, \gamma^*_i,\lambda^*_i, \nu^*, \mu^*$ are the same on both the sides of the inequality. On the other hand, 
\begin{align}
    &L(d^*)- L(\tilde{d})\nonumber\\
    &= R(d_u^*,p_u^*,q_u)+R(d_v^*,p_v^*,q_v)-\gamma^*_ud_u^*\nonumber\\
    &\mkern 40 mu -R(\tilde{d}_u,p_u^*,q_u)-R(\tilde{d}_v,p_v^*,q_v)+\gamma^*_u\tilde{d}_u\label{eq:ldiff:1}\\
    & = R(0,p_u^*,q_u)-R(\epsilon,p_u^*,q_u) +R(d_v^*,p_v^*,q_v)\nonumber\\
    &\mkern 40 mu -R(d_v^*-\epsilon,p_v^*,q_v)+\gamma^*_u\epsilon\nonumber\\
    & \geq h_2(q_u)- \left[3h_2(q_u) -h_3\left(\frac{\epsilon}{2}, q_u\right)  - h_3\left(\frac{\epsilon}{2}, 1-q_u\right)\right]\nonumber\\
    &\mkern 40 mu +R(d_v^*,p_v^*,q_v) -R(d_v^*-\epsilon,p_v^*,q_v)+\gamma^*_u\epsilon\label{eq:ldiff:2}\\
    & \geq -2h_2(q_u)+ h_3\left(\frac{\epsilon}{2}, q_u\right)  + h_3\left(\frac{\epsilon}{2}, 1-q_u\right) - O(\epsilon)\label{eq:ldiff:3}\\
    &>0,\label{eq:ldiff:4}
\end{align}
where \eqref{eq:ldiff:1} uses the complementary slackness condition $\gamma_v^*=0$; in \eqref{eq:ldiff:2}, we use the monotonicity of the RDP function of a Bernoulli source, i.e, $R(\epsilon,p_u^*,q_u) \leq R(\epsilon,0,q_u)$  and its definition; in \eqref{eq:ldiff:3}, we use the inequality $\gamma^*_u\epsilon\geq 0$ and the fact that when $d_v^*>0$, $R(d_v^*,p_v^*,q_v) -R(d_v^*-\epsilon,p_v^*,q_v) \geq \frac{\epsilon}{d_v^*}[R(d_v^*,p_v^*,q_v) -R(0,p_v^*,q_v)]  = -\ O(\epsilon)$ as $\epsilon \to 0$ because of the convexity and monotonicity of the RDP function of a Bernoulli source; and the last inequality \eqref{eq:ldiff:4} follows from the fact that $-2h_2(q_u)+ h_3\left(\frac{\epsilon}{2}, q_u\right)  + h_3\left(\frac{\epsilon}{2}, 1-q_u\right)=\omega(\epsilon)$ as $\epsilon \to 0$. As $L(d^*) > L(\tilde{d})$ contradicts that optimality of $d^*$, we must have $d_i^*>0$ for every $i$.

This along with the complementary slackness condition yields $\gamma_i^*=0$ for every $i$. Therefore, we can omit the term $\sum_{i}\gamma_{i}d_i$ in the Lagrangian function \eqref{eq:subdiff} and consider the function when $d_i>0$ for $i\in [n]$. Let us consider the corresponding KKT conditions, which are given by
 \begin{align}
 &\mathbf{0} \in \partial L|_{d_{i}= d^{*}_i, p_{i}=p
^{*}_i}, \quad i \in [n] \label{eq:lag_der_2}\\
       &d^{*}_i > 0, \quad p^*_{i} \geq 0,  \quad \lambda^*_i \geq 0, \quad i \in [n], \label{eq:d_p_l_2} \\   
       &\sum_{i}d^*_{i} =D,\label{eq:sum_d_2}\\
       &\sum_{i}p^*_{i} = P,\label{eq:sum_p_2}\\
&\lambda^*_i p^{*}_i= 0, \quad i \in [n], \label{eq:lp_2}\\
       &\nu^*, \mu^* \in \mathbb{R},\label{eq:nu_mu_2}
         \end{align}
where $\partial L$ is the subdifferential\footnote{The subdifferential $\partial f$ of a convex function $f(x)$ at $x=x_0$ is defined as $$\partial f|_{x=x_0}=\left\{v \mid f(x_0)+ \langle v, y-x_0\rangle \leq f(y) \ \forall y \in \mathbf{dom}\ f \right\}.$$ We know that for a convex and differentiable $f$, $x=x_0$ is a minimizer of $f$ iff the gradient of $f$ at $x_0$ is the zero vector, i.e., $\nabla f|_{x=x_0}= \mathbf{0}$. In the case of non-differentiable $f$, the necessary and sufficient condition for $x_0$ to be the minimizer of $f$ is that $\mathbf{0} \in \partial f|_{x=x_0}$.} of the function $L$ in the variables $d_{i}$ and $p_{i}$ with the rest of the variables fixed.

As the optimization problem is convex and satisfies Slater's condition, the KKT conditions are the necessary and sufficient conditions for the optimality of the primal and dual problems. It means that the optimizers $\mu^*$, $\nu^*$, $\lambda_i^*, d^{*}_i$ and $p^{*}_i$ for $i \in [n]$ satisfy the above KKT conditions, and any $\mu^*$, $\nu^*$, $\lambda_i^*,  d^{*}_i$ and $p^{*}_i$ for $i \in [n]$ satisfying the KKT conditions are optimizers of the primal and the dual problem. To solve \eqref{eq:lag_der_2}, we will partition the set 
$d_{i}> 0, p_{i}\geq 0$ and consider the subdifferential of $L$ on each part. For each $i$, define 

\begin{align}
    \mathcal{S}_i &\triangleq \Big\{(d_{i},p_{i}) : 0< d_{i}<q_{i},\  \frac{d_{i}(1-2q_{i})}{1-2d_{i}}\leq p_{i}\Big\},\\
     \mathcal{T}_i &\triangleq \left\{(d_{i},p_{i}) : q_{i}<d_{i},\ 2q_{i}(1-q_{i})-(1-2q_{i})p_{i}\leq d_{i},\right. \nonumber\\
    &\mkern 300mu\left.\  0\leq p_{i}  
     \right\},\\
     \mathcal{U}_i &\triangleq \Big\{(d_{i},p_{i}) : \frac{p_{i}}{1-2(q_{i}-p_{i})}<d_{i},\  0\leq p_{i}<q_{i},\nonumber\\
     & \mkern 130mud_i<2q_{i}(1-q_{i})-(1-2q_{i})p_{i}\Big\},\\
    \mathcal{V}_i &\triangleq \left\{(d_{i},p_{i}) : d_{i}=q_{i},\ q_{i}\leq p_{i} \right\},
\end{align}
which is illustrated in Fig.~\ref{fig:pi}. The subdifferential of the Lagrangian function \eqref{eq:subdiff} is given in Table~\ref{table:subdiff}. 
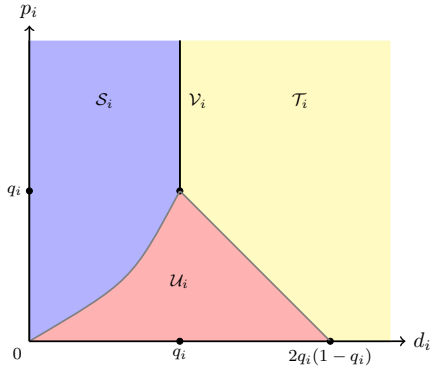
\begin{figure}[h]
\centering
\resizebox{0.8\width}{!}{\input{Figures/pij_dij}}
\caption{A partition of $d_{i}\geq 0$ and $p_{i}\geq 0$}.
\label{fig:pi}
 \end{figure}

\begin{table*}[!h]
\begin{center}
\caption{The subdifferential of the Lagrangian function at $(d^*_{i}, p^*_{i})$ satisfying the conditions \eqref{eq:d_p_l_2} -- \eqref{eq:nu_mu_2}}
\label{table:subdiff}
\begin{tabular}{ |c|c| } 
  \hline
   $(d^*_{i}, p^*_{i})$ & $\partial L|_{d_{i}=d^*_{i}, p_{i}=p^*_{i}}$\\ 
   \hline   $\mathcal{S}_i$&$\left\{\left(\log \frac{d^*_{i}}{1-d^*_{i}}+\nu^*,\  \mu ^*\right)\right\}$\\
   \hline
    $\mathcal{T}_i$ & $\left\{(\nu^*,\  \mu^*-\lambda^*_i)\right\}$\\
    \hline
    $\mathcal{U}_i$ &$\left\{\left(\frac{1}{2}\log \frac{(d^*_{i}-p^*_{i})(d^*_{i}+p^*_{i})}{\left(2(1-q_{i})-(d^*_{i}-p^*_{i})\right)\left(2q_{i}-(d^*_{i}+p^*_{i})\right)}+\nu^*,\  \frac{1}{2}\log \frac{(q_{i}-p^*_{i})^2(d^*_{i}+p^*_{i})\left(2(1-q_{i})-(d^*_{i}-p^*_{i})\right)}{(1-q_{i}+p^*_{i})^2(d^*_{i}-p^*_{i})\left(2q_{i}-(d^*_{i}+p^*_{i})\right)}+\mu^*-\lambda^*_
    {i}\right)\right\}$\\
    \hline
    $\mathcal{V}_i$ & $\operatorname{conv}\left(\left\{\left(\log \frac{q_{i}}{1-q_{i}}+\nu^*, \ \mu^*\right), \  (\nu^*,\  \mu^*)\right\}\right)$\\
    \hline
\end{tabular}
\end{center}
\end{table*}

\begin{lemma}\label{lem:condition_optimizers}
    Let $(d^*_{i}, p^*_{i})$,  $\lambda^*_i$,  $\nu^*$, and $\mu^*$, $i\in [n]$ be the optimizers that satisfy the KKT conditions \eqref{eq:lag_der_2} -- \eqref{eq:nu_mu_2}. Then, only one of the following statements hold.
    \begin{enumerate}
        \item Every pair $(d^*_{i}, p^*_{i})$ must be in the region $\mathcal{S}_i$ or in $\mathcal{V}_i$, i.e.,        
        $$(d^*_{i}, p^*_{i}) \in \mathcal{S}_i \cup \mathcal{V}_i \text{ for } i \in [n].$$
        \item Every pair $(d^*_{i}, p^*_{i})$ must be in the region $\mathcal{T}_i$ or in $\mathcal{V}_i$, i.e.,        
        $$(d^*_{i}, p^*_{i}) \in \mathcal{T}_i \cup \mathcal{V}_i \text{ for } i \in [n].$$
        \item Every pair $(d^*_{i}, p^*_{i})$ must be in the region $\mathcal{U}_i$, i.e.,         
        $$(d^*_{i}, p^*_{i}) \in \mathcal{U}_i \text{ for } i \in [n].$$
    \end{enumerate}
\end{lemma}
\begin{proof}
    The idea is to argue that if one of the pairs $(d^*_{i}, p^*_{i})$'s is in a certain region, which could be $\mathcal{S}_i$, $\mathcal{T}_i$, $\mathcal{U}_i$, or $\mathcal{V}_i$, then the rest of the pairs can only be in a compatible region.

    Let us start with the case of  $(d^*_{i}, p^*_{i}) \in \mathcal{U}_i$ for some $i \in [n]$. As $\mathbf{0} \in \partial L|_{d_{i}=d^*_{i}, p_{i}=p^*_{i}}$, we have 
    \begin{align}
        \nu^*&=-\frac{1}{2}\log \frac{(d^*_{i}-p^*_{i})(d^*_{i}+p^*_{i})}{\left(2(1-q_{i})-(d^*_{i}-p^*_{i})\right)\left(2q_{i}-(d^*_{i}+p^*_{i})\right)}\label{eq:rhs1}\\
        \mu^*-\lambda^*_i&=-\frac{1}{2}\log \frac{(q_{i}-p^*_{i})^2(d^*_{i}+p^*_{i})\left(2(1-q_{i})-(d^*_{i}-p^*_{i})\right)}{(1-q_{i}+p^*_{i})^2(d^*_{i}-p^*_{i})\left(2q_{i}-(d^*_{i}+p^*_{i})\right)}\label{eq:rhs2}
    \end{align}
    As the right-hand sides of \eqref{eq:rhs1} and \eqref{eq:rhs2} are strictly positive\footnote{It can be easily verified that $$\frac{(d_{i}-p_{i})(d_{i}+p_{i})}{\left(2(1-q_{i})-(d_{i}-p_{i})\right)\left(2q_{i}-(d_{i}+p_{i})\right)}<1$$ in $\mathcal{U}_i$ by using the constraint that $d_{i}<2q_{i}(1-q_{i})-(1-2q_{i})p_{i}$. In order to prove that 
    $$\log \frac{(q_{i}-p_{i})^2(d_{i}+p_{i})\left(2(1-q_{i})-(d_{i}-p_{i})\right)}{(1-q_{i}+p_{i})^2(d_{i}-p_{i})\left(2q_{i}-(d_{i}+p_{i})\right)}<0$$ in $\mathcal{U}_i$, we can use convexity of the left-hand side in terms of $d_{i}$ with $p_{i}>0$. To this end, let 
    $$f(d_{i}, p_{i})\triangleq \log \frac{(q_{i}-p_{i})^2}{(1-q_{i}+p_{i})^2}+\log \frac{d_{i}+p_{i}}{d_{i}-p_{i}}+\log \frac{2(1-q_{i})-(d_{i}-p_{i})}{2q_{i}-(d_{i}+p_{i})}.$$
    In $\mathcal{U}_i$ for $p_i>0$, $\frac{p_{i}}{1-2(q_{i}-p_{i})}<d_{i}<2q_{i}(1-q_{i})-(1-2q_{i})p_{i}$ and $f\left(\frac{p_{i}}{1-2(q_{i}-p_{i})}, p_{i}\right)=f\left(2q_{i}(1-q_{i})-(1-2q_{i})p_{i}, p_{i}\right)=0$. Similarly, for $p_i=0$, $f\left(d_{i}, 0\right)= \log\frac{q_i}{1-q_i}<0, f\left(2q_{i}(1-q_{i}), 0\right)=0$.
    As $f(d_{i}, p_{i})$ is strictly convex in the argument $d_{i}$ for any fixed $0\leq p_{i}<q_{i}$, the maximum of $f(d_{i}, p_{i})$ is achieved only at the end points. 
    Hence, we have the required non-positivity of the function.} in the region $\mathcal{U}_i$ and $\lambda^*_i\geq 0$, we have $\nu^*>0$ and $\mu^*>0$. As a result, for a different index $j$, $\mathbf{0} \notin \partial L|_{d_j=d^*_{j}, p_{j}=p^*_{j}}$ when $(d^*_{j}, p^*_{j})$ is in either $\mathcal{S}_{j}, \mathcal{T}_{j},$ or $\mathcal{V}_{j}$. Therefore, $(d^*_{j}, p^*_{j}) \in \mathcal{U}_{j}$, proving statement $3)$.

    We now turn to the case of $(d^*_{i}, p^*_{i}) \in \mathcal{S}_i$ for some $i \in [n]$. As $\mathbf{0} \in \partial L|_{d_{i}=d^*_{i}, p_{i}=p^*_{i}}$, we have
    \begin{align}
        \nu^*&=-\log \frac{d^*_{i}}{1-d^*_{i}}\label{eq:rhs3}\\
        \mu^*&=0\label{eq:rhs4}.
    \end{align}
    Since $d_i^* <q_i \leq \frac{1}{2}$ in $\mathcal{S}_i$, we have $-\log \frac{d^*_{i}}{1-d^*_{i}}>0$, which implies $\nu^*> 0$ and $\mu^*=0$. This rules out the possibility of having $(d^*_{j}, p^*_{j})$ in $\mathcal{U}_{j}$ or $\mathcal{T}_{j}$  for a different index $j$. This is because the subdifferential cannot contain $\mathbf{0}$ as $\lambda_j^*\geq 0$ and $\frac{1}{2}\log \frac{(q_{j}-p^*_{j})^2(d^*_{j}+p^*_{j})\left(2(1-q_{j})-(d^*_{j}-p^*_{j})\right)}{(1-q_{j}+p^*_{j})^2(d^*_{j}-p^*_{j})\left(2q_{j}-(d^*_{j}+p^*_{j})\right)}<0$ in $\mathcal{U}_{j}$. So, the rest of the pairs $(d^*_{j}, p^*_{j})$ must be in $\mathcal{S}_{j} \cup \mathcal{V}_{j}$. 

    Similarly, when $(d^*_{i}, p^*_{i}) \in \mathcal{T}_i$ for some $i \in [n]$, $\nu^*=0$ and $\mu^*=\lambda^*_i$. By combining this with the fact the $\log \frac{d^*_{j}}{1-d^*_{j}}<0$ in $\mathcal{S}_{j}$ and $\frac{1}{2}\log \frac{(d^*_{j}-p^*_{j})(d^*_{j}+p^*_{j})}{\left(2(1-q_{j})-(d^*_{j}-p^*_{j})\right)\left(2q_{j}-(d^*_{j}+p^*_{j})\right)}<0$ in $\mathcal{U}_{j}$, we  can conclude that $(d^*_{j}, p^*_{j})$ cannot be in $\mathcal{S}_{j}$ or $\mathcal{U}_{j}$  for a different index $j$. So it must be in $\mathcal{T}_{j} \cup \mathcal{V}_{j}$.

    Suppose if $(d^*_{i}, p^*_{i}) \in \mathcal{V}_i$ for some $i \in [n]$, then  $0 \leq \nu^* \leq -\log \frac{q_{i}}{1-q_{i}}$ and $\mu^*=0$, which implies that the rest of the pairs can only be in either $\mathcal{S}_{j}, \mathcal{T}_{j},$ or $\mathcal{V}_{j}$. By the above arguments for $\mathcal{S}$ and $\mathcal{T}$, if one of them lies in $\mathcal{S}_{j}$ (resp. $\mathcal{T}_{j}$), then the rest of them must be in the corresponding  $\mathcal{S} \cup \mathcal{V}$ (resp. $\mathcal{T} \cup \mathcal{V}$), which completes the proof of the statements $1)$ and $2)$.
\end{proof}

Let us now complete the proof of  Theorem~\ref{thm:closed_form} by finding an optimizer that satisfies the conditions in Lemma~\ref{lem:condition_optimizers} along with the KKT conditions. First we will consider the case of $(D,P)\in \mathcal{A}$.  As we assume $D>0$, we can set $$d^*_{i}=\min\left\{\frac{1}{e^{\beta}+1}, q_{i}\right\}$$
with $\beta$ chosen such that $\sum_{i} d^*_{i}=D$, and pick $p^*_{i}\geq \frac{2d^*_{i}(1-q_{i})}{1-2d^*_{i}}$ that satisfy $\sum_{i} p^*_{i}=P$. This is possible as $(D,P)\in \mathcal{A}$ where $P\geq T(D)=\sum_{i}\frac{2d^*_{i}(1-q_{i})}{1-2d^*_{i}}$. Observe that  $(d_i^*, p_i^*) \in \mathcal{S}_i \cup \mathcal{V}_i$, and they satisfy the KKT conditions \eqref{eq:lag_der_2}--\eqref{eq:nu_mu_2} with $\lambda_i^*=0$, $\nu^*=\beta$ and $\mu^*=0$. As the KKT conditions  are necessary, we have that  
 \begin{align}
             R(D,P)&=  \sum_{i}R(d_i^*, p_i^*, q_i)\\
             &=\sum_{i}\left[h_2(q_{i}) - h_2(d^*_{i})\right].
\end{align}
for $(D,P)\in \mathcal{A}$. This proves $1)$ of Theorem~\ref{thm:closed_form}.

Next, we will consider the case of $(D,P)\in \mathcal{B}$. In $\mathcal{B}$, we have that $P \geq S(D)$, which means that there exists $(d_i^*, p_i^*)$'s satisfying  $q_{i}\leq d^*_{i},\  0\leq p^*_{i},\ 2q_{i}(1-q_{i})-(1-2q_{i})p^*_{i}\leq d^*_{i},$ $\sum_{i} d^*_{i}=D$ and 
$\sum_{i} p^*_{i}=P$. In other word, $(d_i^*, p_i^*) \in \mathcal{T}_i \cup \mathcal{V}_i$ for $i \in [n]$. This immediately implies that  
$$R(D,P)= \sum_{i}R(d_i^*, p_i^*, q_i)=0$$
for $(D,P)\in \mathcal{B}$, completing the proof of $2)$ of Theorem~\ref{thm:closed_form}.

Before we proceed with the proof of the last statement of the theorem, observe that the optimizers $(d^*_i, p^*_i)\in \mathcal{S}_i \cup \mathcal{V}_i$ for $i \in [n]$ iff they can be written as
\begin{align} \label{eq:d:1}
    d^*_{i}&=\min\left\{\frac{1}{e^{\beta}+1}, q_{i}\right\}\\
p^*_{i}& \geq \frac{2d^*_{i}(1-q_{i})}{1-2d^*_{i}}\label{eq:p:1}
\end{align}
for some $\beta$ chosen such that $\sum_{i} d^*_{i}=D$ and $\sum_{i} p^*_{i}=P$. This follows from equating the subdifferential on $\mathcal{S}_i \cup \mathcal{V}_i$ to zero. Similarly, the condition that the optimizers $(d^*_i, p^*_i)\in \mathcal{T}_i \cup \mathcal{V}_i$ for $i \in [n]$ is nothing but $q_{i}\leq d^*_{i},\  0\leq p^*_{i},\ 2q_{i}(1-q_{i})-(1-2q_{i})p^*_{i}\leq d^*_{i},$ $\sum_{i} d^*_{i}=D$ and 
$\sum_{i} p^*_{i}=P$. 

Let $(D,P)\in \mathcal{C}$. This means that $P < T(D)$ and $P<S(D)$. As $P < T(D)$, it follows from the definition of $T(D)$ that it is not possible to have $d_i^*$ and $p_i^*$ of the form \eqref{eq:d:1} and \eqref{eq:p:1}, which means that $1)$ of Lemma~\eqref{lem:condition_optimizers} cannot happen. Similarly, since $P<S(D)$, it is not possible for $2)$ of Lemma~\eqref{lem:condition_optimizers} to hold. Therefore, by Lemma~\eqref{lem:condition_optimizers}, the optimizers must satisfy the third condition, which is $(d^*_{i}, p^*_{i}) \in \mathcal{U}_i \text{ for } i \in [n].$ This implies that
 \begin{align}
             R&(D,P)\nonumber\\
             &=  \sum_{i}R(d_i^*, p_i^*, q_i)\nonumber\\
             &=\sum_{i}\left[2h_2(q_{i}) + h_2(q_{i}-p^*_{i}) - h_3\left(\frac{d^*_{i}-p^*_{i}}{2}, q_{i}\right)  \right. \nonumber\\
             &\left. \mkern 162mu- h_3\left(\frac{d^*_{i}+p^*_{i}}{2}, 1-q_{i}\right)\right],
         \end{align}
where $(d^*_{i}, p^*_{i})$ satisfy 
\begin{align*}
    -\frac{1}{2}\log \frac{(d^*_{i}-p^*_{i})(d^*_{i}+p^*_{i})}{\left(2(1-q_{i})-(d^*_{i}-p^*_{i})\right)\left(2q_{i}-(d^*_{i}+p^*_{i})\right)}&=\nu^*,\\ -\frac{1}{2}\log \frac{(q_{i}-p^*_{i})^2(d^*_{i}+p^*_{i})\left(2(1-q_{i})-(d^*_{i}-p^*_{i})\right)}{(1-q_{i}+p^*_{i})^2(d^*_{i}-p^*_{i})\left(2q_{i}-(d^*_{i}+p^*_{i})\right)}&=\mu^*-\lambda^*_
    {i}
\end{align*}
for some $\nu^*, \mu^*$ and $\lambda_i^* \geq 0$ such that $\lambda^*_i p^{*}_i= 0, i \in [n]$. Therefore, a solution with $p^{*}_i>0$ has to satisfy
\begin{align*}
    -\frac{1}{2}\log \frac{(d^*_{i}-p^*_{i})(d^*_{i}+p^*_{i})}{\left(2(1-q_{i})-(d^*_{i}-p^*_{i})\right)\left(2q_{i}-(d^*_{i}+p^*_{i})\right)}&=\nu^*,\\ -\frac{1}{2}\log \frac{(q_{i}-p^*_{i})^2(d^*_{i}+p^*_{i})\left(2(1-q_{i})-(d^*_{i}-p^*_{i})\right)}{(1-q_{i}+p^*_{i})^2(d^*_{i}-p^*_{i})\left(2q_{i}-(d^*_{i}+p^*_{i})\right)}&=\mu^*
\end{align*}
and a solution with $p^{*}_i=0$ must satisfy
\begin{align*}
    -\frac{1}{2}\log \frac{(d^*_{i})^2}{\left(2(1-q_{i})-d^*_{i}\right)\left(2q_{i}-d^*_{i})\right)}&=\nu^*,
\end{align*}
which completes the proof of Theorem~\ref{thm:closed_form}.

%% file: Figures/pij_dij.tex
\begin{tikzpicture}[scale=10]

    \draw[thick, ->] (0,0) -- (0.625,0) node[right] {$d_{i}$};
    \draw[thick, ->] (0,0) -- (0,0.525) node[above] {$p_{i}$};
    
    \fill[yellow!30] (0.25,0.5) -- (0.25,0.25) -- (0.5,0) -- (0.6,0) -- (0.6,0.5)-- cycle;
    \fill[red!30] (0,0) .. controls (0.17,0.1) .. (0.25,0.25)--(0.5,0)--cycle;
    
    \fill[blue!30] (0,0.5)--(0,0) .. controls (0.17,0.1) .. (0.25,0.25)--(0.25,0.5)--cycle;

   \node at (0.45,0.4) {\footnotesize$\mathcal{T}_{i}$};
   \node at (0.125,0.4) {\footnotesize$\mathcal{S}_{i}$};
\node at (0.25,0.1) {\footnotesize$\mathcal{U}_{i}$};

\draw [fill] (0.25,0) circle [radius=.005] node [below] {\footnotesize$q_{i}$};
\draw [fill] (0,0.25) circle [radius=.005] node [left] {\footnotesize$q_{i}$};
\draw [fill] (0.5,0) circle [radius=.005] node [below] {\footnotesize$2q_{i}(1-q_{i})$};
\draw [fill] (0.25,0.25) circle [radius=.005];
\node at (0.28,0.4) {\footnotesize$\mathcal{V}_{i}$};
\node at (-0.02,-0.02){\footnotesize$0$};

    \draw[thick] (0.25,0.5) -- (0.25,0.25);
    \draw[thick, gray] (0.25,0.25)-- (0.5,0);
    \draw[thick] (0.5,0) -- (0.6,0);
    \draw[thick, gray] (0,0) .. controls (0.17,0.1) .. (0.25,0.25);    
    \draw[thick]  (0.5,0)--(0,0);    
    \draw[thick] (0,0.5)--(0,0);
    \draw[thick] (0.25,0.25)--(0.25,0.5);
   
\end{tikzpicture}

%% file: RDP_graphs.tex
\section{Rate-Distortion-Perception Function for Graphs}
In this section, we present the rate-distortion-perception function for graphs, motivated by the growing interest in the compression of graph sources \cite{choi_structural_entropy, Abbe2016GraphClusters, delgosha_universal_compression, mihai2021structural, bustin_lossy, martin_rate_distortion_sbm, lossy_spatial_24}. Although a graph is a visual object, it does not make much sense to talk about the notion of ``perception'' in relation to graphs, because the same graph can be drawn in many different ways that look completely different from each other. However, a constraint on the distance between the probability distributions of the original graph and its reconstructed version would still seem to be prudent. By imposing such a constraint along with the Hamming distortion constraint, we can control the essential properties of the reconstructed graph while keeping it as close as possible to the original graph. These properties may include the degree distribution, motiff distribution etc., which are often needed to be preserved for the downstream tasks on the reconstructed graph. As these properties are functions of the distribution of the graph, by controlling the distance between the distributions of the original and reconstructed graphs, we control the closeness of all the graph properties in general. Though the word ``perception" is not appropriate here, we still use it to refer to the similarity of the graph distributions.

Let $G_n=(V, E)$ be a graph with the vertex set $V=[n]$ containing $n$ vertices and the edge set $E$.  Let $P_{G_n}$ denote the probability distribution of a random $n$-vertex graph source. The reconstructed graph with $n$ vertices is denoted by $\hat{G}_n$. Given a distortion measure $\Delta$ and a function $d$ that measures how far the two distributions defined on graphs are, the   rate-distortion-perception function is given by
\begin{align}
    R(D,P)= & \min_{P_{\hat{G}_n|G_n}} I(G_n;\hat{G}_n) \\ 
    &\suchthat \ \mathbb{E}\left[\Delta(G_n;\hat{G}_n)\right]\leq D, \ d(P_{G_n}, P_{\hat{G}_n})\leq P.\nonumber
\end{align}
\subsection{Inhomogeneous Erd\H{o}s-R\'enyi (ER) model}
In an Erd\H{o}s-R\'enyi (ER) model, a graph is generated by connecting pairs of vertices randomly and independently with edges. Precisely, the distribution of an inhomogeneous ER graph is specified by 
$$P_{G_n}=\prod_{i<j}q_{ij}^{E_{ij}}(1-q_{ij})^{1-E_{ij}},$$
where $q_{ij}$ is the probability of connecting two vertices $i$ and $j$ with an edge and $E_{ij}$ is a binary random variable that specifies if an edge between the vertices $i$ and $j$ is present or not.

As there is independence across the edges, we can view an inhomogeneous ER graph as a Bernoulli vector source. Thus, for the Hamming distortion measure $\Delta(G_n;\hat{G}_n)=\sum_{i<j}\mathds{1}(E_{ij} \neq \hat{E}_{ij})$ and a discrepancy measure of the distribution of the edge random variables of the original graph and the reconstructed graph $d(P_{G_n}, P_{\hat{G}_n})=\sum_{i<j}\left|\mathbb{P}(E_{ij}=1) - \mathbb{P}(\hat{E}_{ij}=1)\right|=\sum_{i<j}\left|q_{ij} - \hat{q}_{ij}\right|$, the RDP function of an ER graph is given by   Theorem~\ref{thm:closed_form}.

%% file: Conclusion.tex
\section{Conclusion}
In this work, we have characterized the rate-distortion-perception function of a Bernoulli vector source. The optimal allocation of distortion and perception levels to each component depends on which of the three  possible regions that $(D,P)$ lies in. We have also introduced the RDP function for graph sources and applied the result of Bernoulli vector sources to the Erd\H{o}s-R\'enyi model. It would be of interest to extend the current results by considering different perception measures and different graph sources.

%% file: Main_File.bbl
\begin{thebibliography}{10}
\providecommand{\url}[1]{#1}
\csname url@samestyle\endcsname
\providecommand{\newblock}{\relax}
\providecommand{\bibinfo}[2]{#2}
\providecommand{\BIBentrySTDinterwordspacing}{\spaceskip=0pt\relax}
\providecommand{\BIBentryALTinterwordstretchfactor}{4}
\providecommand{\BIBentryALTinterwordspacing}{\spaceskip=\fontdimen2\font plus
\BIBentryALTinterwordstretchfactor\fontdimen3\font minus \fontdimen4\font\relax}
\providecommand{\BIBforeignlanguage}[2]{{%
\expandafter\ifx\csname l@#1\endcsname\relax
\typeout{** WARNING: IEEEtran.bst: No hyphenation pattern has been}%
\typeout{** loaded for the language `#1'. Using the pattern for}%
\typeout{** the default language instead.}%
\else
\language=\csname l@#1\endcsname
\fi
#2}}
\providecommand{\BIBdecl}{\relax}
\BIBdecl

\bibitem{blau19}
Y.~Blau and T.~Michaeli, ``Rethinking lossy compression: The rate-distortion-perception tradeoff,'' in \emph{Proc.\ 36th Int.\ Conf.\ Mach.\ Learn.\ (ICML)}, California, USA, Jun. 2019, pp. 675--685.

\bibitem{RDP_function_chen}
J.~Chen, L.~Yu, J.~Wang, W.~Shi, Y.~Ge, and W.~Tong, ``On the rate-distortion-perception function,'' \emph{IEEE J.\ Sel.\ Areas\ Inf.\ Theory}, vol.~3, no.~4, pp. 664--673, Dec. 2022.

\bibitem{theis2021a}
\BIBentryALTinterwordspacing
L.~Theis and A.~B. Wagner, ``A coding theorem for the rate-distortion-perception function,'' in \emph{Neural Compression: From Information Theory to Applications -- Workshop @ Int. Conf. Learn. Represent. (ICLR)}, 2021. [Online]. Available: \url{https://openreview.net/forum?id=BzUaLGtKecs}
\BIBentrySTDinterwordspacing

\bibitem{cheuk18}
C.~T. Li and A.~E. Gamal, ``Strong functional representation lemma and applications to coding theorems,'' \emph{IEEE Trans.\ Inf.\ Theory}, vol.~64, no.~11, pp. 6967--6978, Nov. 2018.

\bibitem{RDP_vec_gaussian}
J.~Qian, S.~Salehkalaibar, J.~Chen, A.~Khisti, W.~Yu, W.~Shi, Y.~Ge, and W.~Tong, ``Rate-distortion-perception tradeoff for gaussian vector sources,'' \emph{IEEE J.\ Sel.\ Areas\ Inf.\ Theory}, pp. 1--1, Nov. 2024.

\bibitem{computation_RDP_gaussian}
G.~Serra, P.~A. Stavrou, and M.~Kountouris, ``On the computation of the gaussian rate–distortion–perception function,'' \emph{IEEE J.\ Sel.\ Areas\ Inf.\ Theory}, vol.~5, pp. 314--330, Mar. 2024.

\bibitem{dror_24}
D.~Freirich, N.~Weinberger, and R.~Meir, in \emph{Proc.\ IEEE Int.\ Symp.\ Inf.\ Theory (ISIT)}, Athens, Greece, 2024, pp. 238--243.

\bibitem{yassine_23}
Y.~Hamdi and D.~Gündüz, ``The rate-distortion-perception trade-off with side information,'' in \emph{Proc.\ IEEE Int.\ Symp.\ Inf.\ Theory (ISIT)}, Taipei, Taiwan, 2023, pp. 1056--1061.

\bibitem{Boyd_Vandenberghe_2004}
S.~Boyd and L.~Vandenberghe, \emph{Convex Optimization}.\hskip 1em plus 0.5em minus 0.4em\relax Cambridge University Press, 2004.

\bibitem{choi_structural_entropy}
Y.~Choi and W.~Szpankowski, ``Compression of graphical structures: Fundamental limits, algorithms, and experiments,'' \emph{IEEE Trans.\ Inf.\ Theory}, vol.~58, no.~2, pp. 620--638, Feb. 2012.

\bibitem{Abbe2016GraphClusters}
E.~Abbe, ``{Graph compression: The effect of clusters},'' in \emph{Proc.\ 54th Annu. Allerton Conf.\ Commun.\ Contr.\ Comput.}, Monticello, IL, USA, Sep. 2016, pp. 1--8.

\bibitem{delgosha_universal_compression}
P.~Delgosha and V.~Anantharam, ``Universal lossless compression of graphical data,'' \emph{IEEE Trans.\ Inf.\ Theory}, vol.~66, no.~11, pp. 6962--6976, Nov. 2020.

\bibitem{mihai2021structural}
M.-A. Badiu and J.~P. Coon, ``Structural complexity of one-dimensional random geometric graphs,'' \emph{IEEE Trans.\ Inf.\ Theory}, vol.~69, no.~2, pp. 794--812, Sep. 2023.

\bibitem{bustin_lossy}
R.~Bustin and O.~Shayevitz, ``On lossy compression of directed graphs,'' \emph{IEEE Trans.\ Inf.\ Theory}, vol.~68, no.~4, pp. 2101--2122, Apr. 2022.

\bibitem{martin_rate_distortion_sbm}
M.~W. Wafula, P.~K. Vippathalla, J.~Coon, and M.-A. Badiu, ``Rate-distortion function of the stochastic block model,'' in \emph{Proc.\ 57th Asilomar Conf.\ Signals Syst.\ Comput.}, Pacific Grove, CA, USA, Oct. 2023, pp. 699--703.

\bibitem{lossy_spatial_24}
P.~K. Vippathalla, M.~W. Wafula, M.-A. Badiu, and J.~P. Coon, ``On the lossy compression of spatial networks,'' in \emph{Proc.\ IEEE Int.\ Symp.\ Inf.\ Theory (ISIT)}, Athens, Greece, Jul. 2024, pp. 416--421.

\end{thebibliography}
